	\documentclass[a4paper]{article}
	\usepackage{authblk}
	 \usepackage{fullpage}
	 \usepackage{setspace}
	 \usepackage[english]{babel}
	 \usepackage{amsmath,amsfonts,amssymb,amsthm}
\usepackage{anyfontsize}
\usepackage{hyperref}
\usepackage[utf8]{inputenc}
\usepackage{booktabs}
\usepackage{enumerate}
\usepackage{adjustbox}
\usepackage[group-separator={,}]{siunitx} %
\usepackage{floatrow} %

\usepackage{tabularx}
\newcolumntype{R}[1]{>{\raggedleft\arraybackslash}p{#1}}
\newcolumntype{C}[1]{>{\centering\arraybackslash}p{#1}} %

\usepackage{tikz}
\usetikzlibrary{decorations.pathreplacing,calc}
\usetikzlibrary{matrix,backgrounds,positioning,arrows}

\usepackage{verbatim} 
\usepackage{wrapfig}
\newcommand{\ChangeNodeLabel}[3]{%
	\node[\csname nodeType#1\endcsname=\csname nodeSize#1\endcsname,#2] at (node#1) {#3};
}
\newcommand{\ChangeNode}[2]{%
	\node[\csname nodeType#1\endcsname=\csname nodeSize#1\endcsname,#2] at (node#1) {\csname nodeLabel#1\endcsname};
}
\tikzstyle{ST}=[%
scale=0.7, transform shape, 
inner sep=1mm, font=\Large,
innernode/.style={rectangle,draw=gray!50,fill=gray!20,thick,minimum width=#1,minimum height=0.5cm,rounded corners=1mm,anchor=south},
innernode/.default=4cm,
leafnode/.style={rectangle,draw=black!50,fill=green!20,thick,minimum width=#1,minimum height=0.5cm,anchor=south},
leafnode/.default=1cm,
trienode/.style={radius=0.1cm,yshift=-0.3cm},
charlabel/.style={anchor=east},
trielabel/.style={right,xshift=0.1cm,yshift=-0.3cm},
mast/.style={fill=blue!30},
]

\usepackage{ifthen}
\newboolean{isJournal} 
\setboolean{isJournal}{true} 

\ifthenelse{\boolean{isJournal}}{%

	\newcommand{\JO}[2][]{#2}

}{%
	\usepackage{etoolbox}

	\newcommand{\JO}[2][]{#1}
}

\RequirePackage{thmtools} %
		\newtheorem{theorem}{Theorem}[section]

\newtheorem{Lemma}[theorem]{Lemma}
\newtheorem{Corollary}[theorem]{Corollary}

 \usepackage[natbib,cref,ruby]{nikis} 

\crefname{Theorem}{thm.}{thms.}
\crefname{Lemma}{lemma}{lemmas}
\Crefname{Theorem}{Theorem}{Theorems}
\Crefname{Lemma}{Lem.}{Lems.}
\Crefname{Corollary}{Cor.}{Cors.}
\crefname{Corollary}{cor.}{cors.}
\Crefname{equation}{Eq.}{Eqs.}
\crefname{equation}{eq.}{eqs.}
\Crefname{figure}{Fig.}{Figs.}
\crefname{figure}{fig.}{figs.}
\Crefname{section}{Sec.}{Secs.}
\crefname{section}{sec.}{secs.}

\newcommand{\todo}[1]{} %
\newcommand{\block}[1]{\noindent\textbf{#1. }}

\newcommand{\Time}[1]{\ensuremath{t_{\textup{#1}}}}

	 \newcommand{\timeLCE}[1][]{\ensuremath{t_{\textup{LCE}}}}
	 \newcommand{\timeSA}[1][]{\ensuremath{t_{\textup{\SA}}}}
\newcommand{\lcp}[1][]{\UnaryOperator[#1]{\mathit{lcp}}}
\newcommand{\lcs}[1][]{\UnaryOperator[#1]{\mathit{lcs}}}
\newcommand{\LCP}         {\instancename{LCP}}
\newcommand{\PLCP}         {\instancename{PLCP}}
\newcommand{\SA}          {\instancename{SA}}
\newcommand{\occ}          {\ensuremath{\textup{occ}}}
\newcommand{\ISA}          {\instancename{ISA}}
\newcommand{\LPF}          {\instancename{LPF}}
\newcommand{\bv}[1]{\ensuremath{B_{\mathup{#1}}}}
\newcommand{\B}[1]{\mathsf{b}(#1)}%

	\usepackage[linesnumbered,ruled,vlined]{algorithm2e}
\SetKwProg{function}{Function}{}{}

\title{Computing All Distinct Squares in Linear Time for Integer Alphabets}
	\author[1]{Hideo Bannai}
\author[1]{Shunsuke Inenaga}
\author[2]{Dominik K\"{o}ppl}
\affil[1]{Department of Informatics, Kyushu University, Japan, {\tt inenaga@inf.kyushu-u.ac.jp, bannai@inf.kyushu-u.ac.jp}}
\affil[2]{Department of Computer Science, TU Dortmund, Germany, {\tt dominik.koeppl@tu-dortmund.de}}

\begin{document}
\maketitle{}

\begin{abstract}
Given a string on an integer alphabet, we present an algorithm that computes
the set of all distinct squares belonging to this string in time linear to the string length.
As an application, we show how to compute the tree topology of the minimal augmented suffix tree in linear time.
Asides from that, we elaborate an algorithm computing the longest previous table in a succinct representation using compressed working space.
\end{abstract}

\section{Introduction}
A square is a string of the form $SS$, where $S$ is some non-empty string.
It is well-known that a string of length~$n$ contains at most $n^2/4$ squares. 
This bound is the number of \emph{all} squares, i.e., we count multiple occurrences of the same square, too.
If we consider the number of all \emph{distinct} squares, i.e., we count \emph{exactly one} occurrence of each square,
then it becomes linear in $n$:
The first linear upper bound was given by \citet{FraenkelS98} who proved that
a string of length $n$ contains at most $2n$ distinct squares.
Later, \citet{Ilie07} showed the slightly improved bound of $2n - \Theta(\lg n)$.
Recently, \citet{DezaFT15} refined this bound to $\gauss{11n/6}$.
In the light of these results one may wonder whether future results will ``converge'' to the upper bound of~$n$:
The \emph{distinct square conjecture}~\cite{FraenkelS98,JonoskaMS14StrongerSquare} is that
a string of length $n$ contains at most $n$ distinct squares;
this number is known to be independent of the alphabet size~\cite{ManeaS15SquareDensity}.
However, there still is a big gap between the best known bound and
the conjecture.
While studying a combinatorial problem like this, it is natural to think about ways to actually compute the exact number.

This article focuses on a computational problem on
distinct squares, namely, we wish to compute (a compact representation of)
the set of all distinct squares in a given string.
\citet{GusfieldS04} tackled this problem with an algorithm running in $\Oh{n \sigma_T}$ time, where $\sigma_T$ denotes the number of different characters
contained in the input text $T$ of length~$n$.
Although its running time is optimal $\Oh{n}$ for a constant alphabet,
it becomes $\Oh{n^2}$ for a large alphabet since $\sigma_T$ can be as large as $\Oh{n}$.

We present an algorithm (\Cref{secDistinctSquares}) that computes this set in \Oh{n} time for a given string of length~$n$ over an integer alphabet of size $n^{\Oh{1}}$.
Like \citeauthor{GusfieldS04}, we can use the computed set to decorate the suffix tree with all squares (\Cref{secDecoration}).
As an application, we provide an algorithm that computes the tree topology of the minimal augmented suffix tree~\cite{ApostolicoP96} in linear time (\Cref{secMAST}).
The fastest known algorithm computing this tree topology takes
$\Oh{n \lg n}$ time~\cite{solvingstringstatisticsologn}.

For our approach, we additionally need the longest previous factor table~\cite{FranekHSX03,CrochemoreI08}.
As a side result of independent interest, we show in \Cref{secLPF} how to store this table in $2n + \oh{n}$ bits, and give an algorithm that computes it using compressed working space.

\section{Definitions}\label{secDefinitions}
Our computational model is the word RAM model with word size $\Om{\lg n}$ for some natural number~$n$.
Let $\Sigma$ denote an integer alphabet of size $\sigma = \abs{\Sigma} = n^{\Oh{1}}$.
An element $w$ in $\Sigma^*$ is called a \intWort{string},
and $\abs{w}$ denotes its length.
We denote the $i$-th character of $w$ with $w[i]$, for $1 \le i \le \abs{w}$.
When $w$ is represented by the concatenation of $x, y, z \in \Sigma^*$, i.e., $w = xyz$,
then $x$, $y$ and $z$ are called a \intWort{prefix}, \intWort{substring} and \intWort{suffix} of $w$, respectively.
For any $1 \leq i \leq j \leq |w|$, let $w[i..j]$ denote
the substring of $w$ that begins at position $i$ and ends at position $j$ in $w$.

The \intWort{longest common prefix (LCP)} of two strings is the longest prefix shared by both strings.
The \intWort{longest common extension (LCE)} query asks for the longest common prefix of two suffixes of the \emph{same} string.
The time for an LCE query is denoted by \timeLCE{}.

A \intWort{factorization} of a string $T$ is a sequence of
non-empty substrings of $T$ such that the concatenations of the substrings is $T$.
Each substring is called a \intWort{factor}.

In the rest of this paper, we take a string $T$ of length $n>0$, and call it \intWort{the text}.
We assume that $T[n] = \$$ is a special character that appears nowhere else in $T$,
so that no suffix of $T$ is a prefix of another suffix of $T$. 
We further assume that $T$ is read-only; accessing a character costs constant time.
We sometimes need the \intWort{reverse} of~$T$, which is given by the concatenation $T[n-1] \cdots T[1] \cdot T[n] = T[n-1] \cdots T[1] \$$.

The \intWort{suffix tree} of $T$ is the tree obtained by compacting the trie of all suffixes of $T$;
it has $n$ leaves and at most $n$ internal nodes.
The leaf corresponding to the $i$-th suffix is labeled with~$i$.
Each edge~$e$ stores a string that is called the \intWort{edge label} of $e$.
The \intWort{string label} of a node $v$ is defined as the concatenation of all edge labels on the path from the root to $v$;
the \intWort{string depth} of a node is the length of its string label.

$\SA$ and $\ISA$ denote the suffix array and the inverse suffix array of $T$, respectively~\cite{manber93suffix}.
The access time to an element of $\SA$ is denoted by \timeSA{}.
$\LCP$ is an array such that $\LCP[i]$ is the length of the longest common prefix of $T[\SA[i]..n]$ and $T[\SA[i-1]..n]$ for $i=2,\ldots,n$.
For our convenience, we define $\LCP[1] := 0$.

A \intWort{range minimum query (RMQ)} asks for the smallest value in an integer array for a given range.
There are data structures that can answer an RMQ on an integer array of length~$n$ in constant time while taking $2n+\oh{n}$ bits of space~\cite{fischer11space}.
An LCE query for the suffixes~$T[s..n]$ and~$T[t..n]$ can be answered with an RMQ data structure on \LCP{} with the range $[\min(\ISA[s],\ISA[t])+1..\max(\ISA[s],\ISA[t])]$ in constant time.

A \intWort{bit vector} is a string on a binary alphabet.
A \intWort{select query} on a bit vector asks the position of the $i$-th \bsq{0} or \bsq{1} in the bit vector.
There is a data structure that can be built in $\Oh{n}$ time with $\Oh{n}$ bits of working space such that it takes $\oh{n}$ bits on top of the bit vector,
and can answer a select query in constant time~\cite{clark96compact}.

We identify occurrences of substrings with their position and length in the text, i.e., if $x$ is a substring of $T$, then there is 
a $1 \le i \le n$ and a $0 \le \ell \le n-i+1$ such that $T[i..i+\ell-1] = x$. 
In the following, we will represent the occurrences of substrings by tuples of position and length.
When storing these tuples in a set, we call the set \intWort{distinct}, if there are no two tuples $(i,\ell)$ and $(i',\ell)$
such that $T[i..i+\ell-1] = T[i'..i'+\ell-1]$.
A special kind of substring is a square:
A \intWort{square} is a string of the form $SS$ for $S \in \Sigma^+$;
we call $S$ and \abs{S} the \intWort{root} and the \intWort{period} of the square $SS$, respectively.
Like with substrings, we can generate a set containing some occurrences of squares.
A set of \intWort{all distinct squares} is a distinct set of occurrences of squares that is maximal under inclusion.

\JO{%
To compute a set of all distinct squares is the main focus of this paper.
We will tackle this problem theoretically in \Cref{secDistinctSquares}, and practically in \Cref{secEval}.
Finally, we give two applications of this problem in \Cref{secDecoration} and \Cref{secMAST}.
But before all that, we start with the study of the LPF array needed for our approach computing all distinct squares:
}%

\section{A Compact Representation of the LPF Array}\label{secLPF}

The longest previous factor table \LPF{} of~$T$ is formally defined as
\[
	\LPF[j] := \max\menge{\ell \mid \text{there exists an~} i \in [1..j-1] \text{~such that~} T[i,i+\ell-1] = T[j,j+\ell-1]}.
\]
It is useful for computing the \intWort{Lempel-Ziv factorization} of~$T = f_1 \cdots f_z$, which is defined as
$f_i = T[k..k+\max(1,\LPF[k])]$ with $k:= \sum_{j=1}^{i-1}\abs{f_{j}}+1$ for $1 \le i \le z$.

In the following, we will use the text
$
T =
\ruby{\texttt{a}}{0}%
\ruby{\texttt{b}}{0}%
\ruby{\texttt{a}}{3}%
\ruby{\texttt{b}}{2}%
\ruby{\texttt{a}}{1}%
\ruby{\texttt{a}}{2}%
\ruby{\texttt{a}}{5}%
\ruby{\texttt{b}}{4}%
\ruby{\texttt{a}}{3}%
\ruby{\texttt{b}}{2}%
\ruby{\texttt{a}}{1}%
\ruby{\texttt{\$}}{0}%
$
as our running example
whose \LPF{} array is represented by the small numbers above the characters.
The Lempel-Ziv factorization of $T$ is given by
$
	\RubyReset{}%
	\RubyCount{a}|%
	\RubyCount{b}|%
	\RubyCount{aba}|%
	\RubyCount{aa}|%
	\RubyCount{baba}|%
	\RubyCount{\$}
$, where the small numbers denote the factor indices, and the vertical bars denote the factor borders.
\JO{\Cref{figRunningEx} shows \SA{}, \LPF{} and other used array data structures of our running example.}{}

\JO{%
\begin{figure}[t]
	\floatbox[{\capbeside\thisfloatsetup{capbesideposition={right},capbesidewidth=3cm}}]{figure}[\FBwidth]
	{%
		\caption{The arrays \SA{}, \LCP{}, \PLCP{} and \LPF{} of the running example.}
		\label{figRunningEx}
}{%
		\renewcommand{\familydefault}{\ttdefault}\normalfont
	\tabcolsep=0.1em
		\begin{tabular}{l*{15}{R{1.8em}}}
				$i$        & 1  & 2  & 3 & 4 & 5 & 6 & 7 & 8 & 9  & 10 & 11 & 12
				\\\toprule
				$T$        & {\tt a}  & {\tt b}  & {\tt a} & {\tt b} & {\tt a} & {\tt a} & {\tt a} & {\tt b} & {\tt a}  & {\tt b}  & {\tt a}  & {\tt \$}
				\\\midrule
				$\SA$      & 12 & 11 & 5 & 6 & 9 & 3 & 7 & 1 & 10 & 4  & 8  & 2
				\\\midrule
				$\LCP$     & 0  & 0  & 1 & 2 & 1 & 3 & 3 & 5 & 0  & 2  & 2  & 4
				\\\midrule
				$\PLCP$    & 5  & 4  & 3 & 2 & 1 & 2 & 3 & 2 & 1  & 0  & 0  & 0
				\\\midrule
				$\LPF$     & 0  & 0  & 3 & 2 & 1 & 2 & 5 & 4 & 3  & 2  & 1  & 0
				\\\bottomrule
			\end{tabular}
}%
\end{figure}
}%

\begin{Corollary}\label{corLZ}
	Given \LPF{}, we can compute the Lempel-Ziv factorization in \Oh{n} time.
	If the factorization consists of $z$ factors,
	the factorization can be represented by an array of $z \lg n$ bits, 
	where the $x$-th entry stores the beginning of the $x$-th factor.

	Alternatively, it can be represented by a bit vector of length~$n$ in which we mark the factor beginnings.
	A select data structure on top of the bit vector can return the length and the position of a factor in constant time.
\end{Corollary}

Since we will need \LPF{} in \Cref{secDistinctSquares}, we are interested in the time and space bounds for computing \LPF{}.
We start with the (to the best of our knowledge) state of the art algorithm with respect to time and space requirements.
\begin{Lemma}[{\cite[Theorem 1]{Crochemore2009LCR}}]\label{lemLPFCrochemore}
Given \SA{} and \LCP{},
we can compute \LPF{} in \Oh{n \timeSA} time.
Besides the output space of $n \lg n$ bits, we only need constant working space.
\end{Lemma}
Apart from this algorithm, we are only aware of some practical improvements~\cite{OhlebuschG11,KarkkainenKP13}.

Let us consider the size of \LCP{} needed in \Cref{lemLPFCrochemore}.
\citet{cstsada} showed a $2n + \oh{n}$-bits representation of \LCP{}.
Thereto he stores the array $\PLCP$ defined as $\PLCP[\SA[i]] = \LCP[i]$ in a bit vector in the following way (also described in~\cite{fischer10wee}):
Since
$\PLCP[1]+1, \PLCP[2]+2,\ldots,\PLCP[n]+n$ is a non-decreasing sequence with $1 \le \PLCP[1]+1 \le \PLCP[n]+n = n$
($\PLCP[i] \le n-i$ since the terminal {\tt \$} is a unique character in~$T$)
the values $I[1] := \PLCP[1]$ and $I[i] := \PLCP[i] - \PLCP[i-1]+1$ ($2 \le i \le n$) are non-negative.
By writing $I[i]$ in the unary code
${\tt 0}^{I[i]} {\tt 1}$ to a bit vector $S$ subsequently for each $2 \le i \le n$,
we can compute
$\PLCP[i] = \select[1](S, i) - 2 i$
and
$\LCP[i] = \select[1](S, \SA[i]) - 2 \SA[i]$.
Moreover, $\sum_{i=1}^n I[i] \le n$ and therefore $S$ is of length at most $2n$.

By using \citeauthor{cstsada}'s \LCP{}-representation,
we get \LPF{} with the algorithm of \citet{Crochemore2009LCR} in the following time and space bounds:
\begin{Corollary}\label{corLPFsmall}
	Having \SA{} and \LCP{} stored in $n \lg n$ bits (this allows \timeSA = \Oh{1}) and $2n + \oh{n}$ bits, respectively, 
we can compute \LPF{} with $\Oh{\lg n}$ additional bits of working space (not counting the space for \LPF{}) in \Oh{n} time.
\end{Corollary}
By plugging in a suffix array construction algorithm like the in-place construction algorithm by \citet{LiLH16a}, we get 
the bounds shown in \Cref{tableLPF} (since we can build \LCP{} in-place after having \SA{}~\cite{hon02space}).

Although this result seems compelling, this approach stores \SA{} and \LPF{} in plain arrays (the former for getting constant time access).
In the following, we will show that the \LPF{} array can be stored more compactly.
We start with a new representation of \LPF{}, for which we use the same trick as for \PLCP{} due to the following property
(which is crucial for squeezing \PLCP{} into $2n + \oh{n}$ bits).

\begin{Lemma}\label{lemLPFprop}
\(
n-j \ge
\LPF[j] \ge \LPF[j-1]-1
\)
for $2 \le j \le n$.
\end{Lemma}
\begin{proof}
There is an $i$ with $1 \le i < j-1$ such that $T[i..i+\LPF[j-1]-1] = T[j-1..j-1+\LPF[j-1]-1]$.
Hence $T[i+1..i+\LPF[j-1]-1] = T[j..j-1+\LPF[j-1]-1]$.
\end{proof}
We conclude that the sequence $\LPF[1]+1, \LPF[2]+2,\ldots,\LPF[n]+n$ is non-decreasing with $1 \le \LPF[1]+1 \le \LPF[n]+n \le n$.
We immediately get:

\begin{Corollary}\label{corLPFrepresentation}
	\LPF{} can be represented by a bit vector with a select data structure such that accessing an \LPF{} value can be performed in constant time.
	The data structures use $2n + \oh{n}$ bits.
\end{Corollary}

To get a better working space bound, 
we have to come up with a new algorithm since the algorithm of \Cref{lemLPFCrochemore} creates a plain array to get constant time random write-access for computing the entries of \LPF{}.
To this end, we present two algorithms that compute \LPF{} in this representation with the aid of the suffix tree.
The two algorithms are derivatives of the algorithms~\cite{lzcics,lzciss} that compute the Lempel-Ziv factorization, either 
in \Oh{n \lg \lg \sigma} time using \Oh{n \lg \sigma} bits, or in \Oh{n/\epsilon^2} time using $(1+\epsilon)n \lg n + \Oh{n}$ bits, for a constant $0 < \epsilon \le 1$.
The current bottleneck of both algorithms is the suffix tree implementation with respect to space and time.
Due to current achievements~\cite{mnk17,LiLH16a}, the algorithms now run in \Oh{n} time using \Oh{n \lg \sigma} bits, or in \Oh{n/\epsilon} time using $(1+\epsilon)n \lg n + \Oh{n}$ bits, respectively.

\begin{table}[t]
	\floatbox[{\capbeside\thisfloatsetup{capbesideposition={right},capbesidewidth=3.5cm}}]{figure}[\FBwidth]
	{%
	\caption{Algorithms computing \LPF{}; space is counted in bits. The output space~$\abs{\LPF}$ is not considered as working space. $0 < \epsilon \le 1$ is a constant.}
	\label{tableLPF}
}{%
\hspace{-1em}
		\begin{tabular}{llll}
			algorithm & time & working space & \abs{\LPF}
			\\\toprule
			\Cref{lemLPFCrochemore},\cite{Crochemore2009LCR} & ${\Oh{n \timeSA}}$ & ${\abs{\SA} + \abs{\LCP} + \Oh{\lg n}}$ & $n \lg n$ 
		\\\midrule
		\Cref{corLPFsmall},\cite{LiLH16a,hon02space} & ${\Oh{n}}$ & ${n \lg n + 2n + \Oh{\lg n}}$ & $n \lg n$ 
			 \\\midrule
				\Cref{lemLPFmy},\cite{lzcics} & ${\Oh{n/\epsilon}}$ & ${(1+\epsilon)n \lg n + \Oh{n}}$ & ${2n + \oh{n}}$ 
			\\\midrule
			\Cref{lemLPFmy},\cite{lzciss} & ${\Oh{n \timeSA}}$ & ${\Oh{n \lg \sigma}}$ & ${2n + \oh{n}}$ 
			 \\\bottomrule
		\end{tabular}
}%
\end{table}

We aim at building the \LPF{}-representation of \Cref{corLPFrepresentation} directly such that we do not need to allocate 
the plain LPF array using $n \lg n$ bits in the first place.
To this end we create a bit vector of length $2n$ and store the \LPF{} values in it successively.
In more detail, we follow the description of the Lempel-Ziv factorization algorithms presented in~\cite{lzcics,lzciss}.
There, the algorithms are divided into several passes.
In each pass we successively visit leaves in text order (determined by the labels of the leaves).
To compute \LPF{}, we only have to do a single pass.
Similarly to the first passes of the two Lempel-Ziv algorithms, 
we use a bit vector~\bv{V} to mark already visited internal nodes.
On visiting a leaf we climb up the tree until reaching the root or an already marked node.
In the former case (we climbed up to the root) we output zero.
In the latter case, we output the string depth of the marked node.
By doing so, we have computed $\LPF[1..j]$ after having processed the leaf with label~$j$.

\begin{Lemma}\label{lemLPFmy}
	We can compute \LPF{} in \Oh{n \timeSA} time with \Oh{n \lg \sigma} bits of
working space, or in \Oh{n/\epsilon} time using
$(1+\epsilon)n \lg n + \Oh{n}$ bits of working space, for a constant $0 < \epsilon \le 1$.
Both variants include the space of the output in their working spaces.
\end{Lemma}
\begin{proof}
Computing the string depth of a node needs access to an RMQ data structure of \LCP{}, and an access to \SA{}.
Both accesses can be emulated by the compressed suffix array in $\timeSA$ time, given that we have computed \PLCP{} in the 
above representation.
\end{proof}

\section{Computing the Set of All Distinct Squares}\label{secDistinctSquares}
Given a string~$T$, our goal is to compute all distinct squares of~$T$.
Thereto we return a set of pairs, where each pair~$(s,\ell)$ consists of a starting position~$s$ and a length~$\ell$ such that 
$T[s..s+\ell-1]$ is the leftmost occurrence of a square.
The size of this set is linear due to
\begin{Lemma}[\citet{FraenkelS98}]\label{lemNumberSquaresLinear}
	A string of length~$n$ can contain at most $2n$ distinct squares.
\end{Lemma}

We follow the approach of~\citet{GusfieldS04}.
Their idea is to compute a set of squares (the set stores pairs of position and length like described in \Cref{secDefinitions})\footnote{It differs to the set we want to compute by the fact that they allow, among others, occurrences of the same square in their set.} with which they can generate all distinct squares.
They call this set of squares a \intWort{leftmost covering set}.
A leftmost covering set obeys the property that every square of the text can be constructed by right-rotating a square of this set.
A square~$(k,\ell)$ is constructed by \intWort{right-rotating} a square~$(i,\ell)$ with $i \le k$ iff
each tuple~$(i+j,\ell)$ with $1 \le j \le k-i$ represents a square~$T[i+j..i+\ell+j-1] = T[i+j..i+\ell-1]T[i..i+j-1]$.

The set of the leftmost occurrences of all squares is a set of all distinct squares.
Unfortunately, the leftmost covering set computed in~\cite{GusfieldS04} is not necessarily a set of all distinct squares since 
(a) it does not have to be distinct, and (b) a square might be missing that can be constructed by right-rotating a square of the computed leftmost covering set.

For illustration, the squares of our running example $T = $
\begin{adjustbox}{trim=0 1em 0 0}
\begin{tikzpicture}
\tikzstyle{array} = [matrix of nodes,
inner sep=0.2mm,
font=\normalsize\ttfamily,nodes={draw, minimum size=0mm},column sep=-\pgflinewidth, row sep=0mm, nodes in empty cells,
	row 1/.style={nodes={draw=none, fill=none,text height=0.5em, minimum size=0mm}},
]
\matrix[array] (warray) {%
a&b&a&b&a&a&a&b&a&b&a&\$
\\};

\tikzstyle{fsquare} = [line width=0.1mm]
\tikzstyle{nsquare} = [densely dotted]

\newcommand{\braU}[4][]{%
 \draw [#1] ([yshift=#4] warray-1-#2.north west) -- node[above] {} ([yshift=#4] warray-1-#3.north east);
 }
\newcommand{\braD}[4][]{%
 \draw [#1] ([yshift=-#4] warray-1-#2.south west) -- node[below] {} ([yshift=-#4] warray-1-#3.south east);
 }

\braU[fsquare]{1}{4}{0.5mm}
\braU[fsquare]{2}{5}{1mm}
\braD[fsquare]{5}{6}{0.5mm}
\braD[nsquare]{6}{7}{1mm}
\braU[nsquare]{7}{10}{0.5mm}
\braU[nsquare]{8}{11}{1mm}
\end{tikzpicture}
\end{adjustbox}
are highlighted with bars.
The set of all squares is
$\menge{(1,4), (2,4), (5,2), (6,2), (7,4), (8,4)}$.
If we take the leftmost occurrences of all squares, we get
\menge{(1,4), (2,4), (5,2)}; this set comprises all squares marked by the solid bars, i.e., the dotted bars correspond to occurrences of squares that are not leftmost. 
In this example, the dotted bars form the set \menge{(6,2), (7,4), (8,4)}, which is a set of all distinct squares.
A leftmost covering set is \menge{(1,4),(5,2)}.

Our goal is to compute the set of all leftmost occurrences directly by modifying the algorithm of~\cite{GusfieldS04}.
To this end, we briefly review how their approach works:
They compute their leftmost covering set by examining the borders between all Lempel-Ziv factors $f_1 \cdots f_z = T$.
That is because of
\begin{Lemma}[{\cite[Theorem 5]{GusfieldS04}}]\label{lemLeftmostOcc}
	The leftmost occurrence of a square~$T[i..i+2p-1]$ touches at least two Lempel-Ziv factors.
	Let $f_x$ be the factor that contains the center of the square~$i+p-1$. Then either
	\begin{enumerate}[(a)]
		\item the square has its left end (position~$i$) inside~$f_x$ and its right end (position $i+2p-1$) inside $f_{x+1}$, or \label{itLeftmostOcc1}
		\item the left end of the square extends into $f_{x-1}$ (or even further left). \label{itLeftmostOcc2}
			The right end can be contained inside~$f_x$ or~$f_{x+1}$.
	\end{enumerate}
\end{Lemma}

\begin{figure}[t]
	\centering{%
		\begin{tikzpicture}[xscale=0.9,yscale=0.5]
\newcount\fX%
\pgfmathsetcount{\fX}{3}
\newcount\fnX%
\pgfmathsetcount{\fnX}{5}
\newcount\qX%
\pgfmathsetcount{\qX}{2}
\newcommand{\Height}{1.7}
\newcommand{\ellOne}{0.5}
\newcommand{\ellTwo}{0.7}
\draw (0,\Height) rectangle node {} (\fX,0);
\draw [<->] (\qX,-0.5) -- node [midway,below] {$p$} (\fX,-0.5);
\draw[decorate,decoration={brace,amplitude=10pt}](\fX,0) -- node[below,yshift=-3mm] {$f_x$} (0,0);
\draw (\fX,\Height) rectangle node {} (\fnX,0);
\draw[decorate,decoration={brace,amplitude=10pt}](\fnX,0) -- node[below,yshift=-3mm] {$f_{x+1}$} (\fX,0);
\draw (-1,0) -- (0,0);
\draw (-1,\Height) -- (0,\Height);
\draw (\fnX,0) -- ($(\fnX,0)+(1,0)$);
\draw (\fnX,\Height) -- ($(\fnX,\Height)+(1,0)$);

\draw [dashed] (\qX,0) -- (\qX,\Height) node [above] {$q$};
\draw [->] (\qX,0.75) -- node [midway,above] {$\ell_{\textup{R}}$} ($(\qX,0.75)+(\ellOne,0)$);
\draw [->] (\fX,0.5) -- node [midway,above] {$\ell_{\textup{R}}$} ($(\fX,0.5)+(\ellOne,0)$);

\draw [<-] ($(\qX,0.5)-(0.0+\ellTwo,0)$) -- node [midway,above] {$\ell_{\textup{L}}$} ($(\qX,0.5)$);
\draw [<-] ($(\fX,0.25)-(\ellTwo,0)$) -- node [midway,above] {$\ell_{\textup{L}}$} ($(\fX,0.25)$);
\end{tikzpicture}
\hspace{1em}
\begin{tikzpicture}[xscale=0.9,yscale=0.5]
\newcount\fX%
\pgfmathsetcount{\fX}{3}
\newcount\fnX%
\pgfmathsetcount{\fnX}{-2}
\newcount\qX%
\pgfmathsetcount{\qX}{1}
\newcommand{\Height}{1.7}
\newcommand{\ellOne}{0.5}
\newcommand{\ellTwo}{0.7}
\draw (0,\Height) rectangle node {} (\fX,0);
\draw [<->] (0,-0.5) -- node [midway,below] {$p$} (\qX,-0.5);
\draw[decorate,decoration={brace,amplitude=10pt}](\fX,0) -- node[below,yshift=-3mm] {$f_x$} (0,0);
\draw (\fX,\Height) rectangle node {} (\fnX,0);
\draw[decorate,decoration={brace,amplitude=10pt}](0,0) -- node[below,yshift=-3mm] {$f_{x-1}$} (\fnX,0);
\draw (\fX,0) -- ($(\fX,0)+(1,0)$);
\draw (\fX,\Height) -- ($(\fX,\Height)+(1,0)$);
\draw (\fnX,0) -- ($(\fnX,0)-(1,0)$);
\draw (\fnX,\Height) -- ($(\fnX,\Height)-(1,0)$);

\draw [dashed] (\qX,0) -- (\qX,\Height) node [above] {$q$};
\draw [->] (\qX,0.25) -- node [midway,above] {$\ell_{\textup{R}}$} ($(\qX,0.25)+(\ellOne,0)$);
\draw [->] (0,0.75) -- node [midway,above] {$\ell_{\textup{R}}$} ($(0,0.75)+(\ellOne,0)$);

\draw [<-] ($(\qX,0.5)-(0.0+\ellTwo,0)$) -- node [midway,above] {$\ell_{\textup{L}}$} ($(\qX,0.5)$);
\draw [<-] ($(0,0.25)-(\ellTwo,0)$) -- node [midway,above] {$\ell_{\textup{L}}$} ($(0,0.25)$);
\end{tikzpicture}
	}
	\caption{Search for squares on Lempel-Ziv borders. 
	The left image corresponds to squares of type \Cref{lemLeftmostOcc}(\ref{itLeftmostOcc1}), the right image to the type \Cref{lemLeftmostOcc}(\ref{itLeftmostOcc2}). 
	Given two adjacent factors, we determine a position~$q$ that is $p$ positions away from the border (the direction is determined by the type of square we want to search for). 
	By two LCE queries we can determine the lengths $\ell_{\textup{L}}$ and $\ell_{\textup{R}}$ that indicate the presence of a square if $\ell_{\textup{L}} + \ell_{\textup{R}} \ge p$.
	}
	\label{figSquareSearch}
\end{figure}
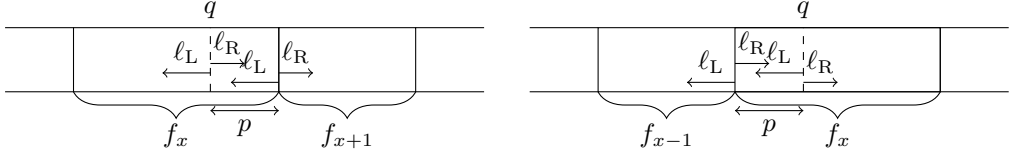

Having a data structure for computing LCE queries on the text and on its inverse, they can probe at the borders of two consecutive factors whether there is a square.
Roughly speaking, they have to check at most $\abs{f_x}+\abs{f_{x+1}}$ many periods at the borders of every two consecutive factors~$f_x$ and $f_{x+1}$ due to the above lemma.
This gives $\sum_{x=1}^z \timeLCE \tuple{ \abs{f_x}+\abs{f_{x+1}} } = \Oh{n \timeLCE}$ time, 
during which they can compute a leftmost covering set~$L$.
\Cref{figSquareSearch} visualizes how the checks are done.
Applying the algorithm on our running example will yield the set~$L = \menge{(1,4), (5,2), (7,4)}$.
To transform this set into a set of all distinct squares, their algorithm runs the so-called Phase~II that uses the suffix tree.
It begins with computing the locations of the squares belonging to a subset $L' \subseteq L$ in the suffix tree in \Oh{n} time.
This subset $L'$ is still guaranteed to be a leftmost covering set.
Finally, their algorithm computes all distinct squares of the text
by right-rotating the squares in $L'$.
In their algorithm, the right-rotations are done by
\emph{suffix link walks} over the suffix tree.
Their running time analysis is based on the fact that
each node has at most $\sigma_T$ incoming suffix links,
where $\sigma_T$ denotes the number of different characters occurring in
the text $T$.
Given that the number of distinct squares is linear, Phase~II runs in $\Oh{n \sigma_T}$ time.

In the following, we will present our modification of the above sketched algorithm.
To speed up the computation, 
we discard the idea of using the suffix links for right-rotating squares (i.e., we skip Phase~II completely).
Instead, we compute a list of all distinct squares directly.
To this end, we show a modification of the sketched algorithm such that it outputs this list sorted first by the lengths (of the squares), and second by the starting position.

First, we want to show that we can change the original algorithm to output its leftmost covering set in the above described order. 
To this end, we iterate over all possible periods, and search not yet reported squares at all Lempel-Ziv borders, for each period.
To achieve linear running time,
we want to skip a factor~$f_x$ when the period becomes longer than~$\abs{f_x}+\abs{f_{x+1}}$.
We can do this with an array~$Z$ of $z \lg z$ bits that is zero initialized.
When the period~$p$ becomes longer than~$\abs{f_x}+\abs{f_{x+1}}$, we write $Z[x] \gets \min \menge{y > x : \abs{f_y}+\abs{f_{y+1}} \ge p}$ such that $Z[x]$ refers to the next
factor whose length is sufficiently large.
By doing so, if $Z[x] \not= 0$, we can skip all factors~$f_y$ with $y \in [x..Z[x]-1]$ in constant time.
This allows us running the modified algorithm still in linear time.

We have to show that the modified algorithm still computes the same set.
To this end, let us fix the period~$p$ (over which we iterate in the outer loop).
By~\cite[Lemma 7]{GusfieldS04}, processing squares satisfying \Cref{lemLeftmostOcc}(\ref{itLeftmostOcc1}) before processing squares satisfying \Cref{lemLeftmostOcc}(\ref{itLeftmostOcc2})
(all squares have the same period~$p$) produces the desired output for period~$p$.

Finally, we show the modification that computes all distinct squares (instead of the original leftmost covering set).
On a high level, we use an RMQ data structure on \LPF{} to filter already found squares.
The filtered squares are used to determine the leftmost occurrences of all squares by right-rotation.
In more detail, we modify Algorithm~1 of~\cite{GusfieldS04} by filtering the squares in the following way~(see \Cref{algo1}):
For each period~$p$, we use a bit vector~\bv{} marking the beginning positions of all found squares with period~$p$.
On reporting a square, we additionally mark its starting position in \bv{}. 
By doing so, an invariant of the algorithm below is that all right-rotated squares of a marked square are already reported.

Let us assume that we are searching for the leftmost occurrences of all squares whose periods are equal to~$p$.
Given the starting position~$s$ of a square returned by~\cite[Algorithm 1]{GusfieldS04},
we consider the square~$(s,2p)$ and its right-rotations as candidates of our list:
If $\bv{}[s] = 1$, then this square and its right-rotations have already been reported.
Otherwise, we report $(s,2p)$ if $\LPF[s] < 2p$.
In order to find the leftmost occurrences of all not yet reported right-rotated squares efficiently, 
we first compute the rightmost position~$e$ of the repetition of period~$p$ containing the square~$(s,2p)$ by an LCE query.
Second, we check the interval $I := [s+1..\min(s+p-1,e-2p+1)]$ for the starting positions of the squares whose \LPF{} values are less than $2p$.
To this end, we perform an RMQ query on \LPF{} to find the position~$j$ whose \LPF{} value is minimal in~$I$.
If $j > 2p$, then all squares with period~$p$ in the considered range have already been found, 
i.e., there is no leftmost occurrence of a square with the period~$p$ in this range.
Otherwise, we report~$(j,2p)$ and recursively search for the text position with the minimal \LPF{} value within the intervals~$[s+1..j-1]$ and~$[j+1..\min(s+p-1,e-2p+1)]$.
In overall, the time of the recursion is bounded by twice of the number of distinct squares starting in the interval~$I$, since a recursion step terminates if it could not report any square.

\begin{theorem}\label{thmDistinctSquares}
	Given an LCE data structure with \timeLCE{} access time and \LPF{},
	we can compute all distinct squares in $\Oh{n t_{\text{LCE}}+ \occ} = \Oh{n t_{\text{LCE}}}$ time,
	where \occ{} is the number of distinct squares.
\end{theorem}
\begin{proof}
We show that the returned list is the list of all distinct squares.
No square occurs in the list twice since we only report the occurrence of a square $(i,\ell)$ if $\LPF[i] < \ell$.
Assume that there is a square missing in the list; let $(i,\ell)$ be its leftmost occurrence.
There is a square~$(j,\ell)$ reported by the (original) algorithm~\cite{GusfieldS04} such that $i - \ell/2 < j \le i$ and right-rotating~$(j,\ell)$ yields~$(i,\ell)$.
Since we right-rotate all found squares, we obviously have reported~$(j,\ell)$.

The \occ{} term in the running time is dominated by the $n t_{\text{LCE}}$ term due to \Cref{lemNumberSquaresLinear}.
\end{proof}

\section{Practical Evaluation}\label{secEval}
We have implemented the algorithm computing the leftmost occurrences of all squares in \CPlusPlus{}11~\cite{repo}.
The primary focus was on the execution time, rather than on a small memory footprint:
We have deliberately chosen plain 32-bit integer arrays for storing all array data structures like \SA{}, \LCP{} and \LPF{}.
These data structures are constructed as follows:
First, we generate \SA{} with {\tt divsufsort}~\cite{divsufsort}.
Subsequently, we generate \LCP{} with the $\Phi$-algorithm~\cite{kaerkkaeinen09permuted},
and \LPF{} with the simple algorithm of~\cite[Proposition 1]{Crochemore2009LCR}. %
Finally, we use the bit vector class and the RMQ data structure provided by the sdsl-lite library~\cite{gbmp2014sea}.
In practice, it makes sense to use an RMQ only for very large LCP values and periods (i.e., RMQs on \LPF{}) due to its long execution time.
For small values, we naively compared characters, or scanned \LPF{} linearly.

We ran the algorithm on all 200MiB collections of the Pizza\&Chili Corpus~\cite{PizzaChili}.
The Pizza\&Chili Corpus is divided in a real text corpus with the prefix \textsc{pc}, and in a repetitive corpus with the prefix \textsc{pcr}.
The experiments were conducted on a machine with 32 GB of RAM and an
Intel\textsuperscript{\textregistered} Xeon\textsuperscript{\textregistered} CPU~\texttt{E3-1271~v3}.
The operating system was a 64-bit version of Ubuntu Linux~\num{14.04} with the kernel version~3.13.
We used a single execution thread for the experiments.
The source code was compiled using the GNU compiler {\tt g++~6.2.0} with the compile flags {\tt -O3 -march=native -DNDEBUG}.

\Cref{tableEval} shows the running times of the algorithm on the described datasets.
It looks like that large factors tend to slow down the computation, since the algorithm has to check all periods up to
$\max_x (\abs{f_x}+\abs{f_{x+1}})$. This seems to have more impact on the running time than the number of Lempel-Ziv factors~$z$.

\begin{table}[t]
	\centerline{%
	{\footnotesize
		\begin{tabular}{l*{8}{r}}
			collection              & $\sigma$  & $\max_i \LCP[i]$   & $\textup{avg}_{\LCP{}}$        & $z$           & $\max_x \abs{f_x}$ & $\max_x \abs{f_xf_{x+1}}$ & $\abs{\occ}$   & time \\
			\toprule
\textsc{pc-dblp.xml    }  & \num{97}  & \num{1084}    & \num{44}     & \num{7035342}  & \num{1060}         & \num{1265}                       & \num{7412}     & \num{70}   \\
\textsc{pc-dna         }  & \num{17}  & \num{97979}   & \num{60}     & \num{13970040} & \num{97966}        & \num{97982}                      & \num{132594}   & \num{310}  \\
\textsc{pc-english     }  & \num{226} & \num{987770}  & \num{9390}   & \num{13971134} & \num{987766}       & \num{1094108}                    & \num{13408}    & \num{2639} \\
\textsc{pc-proteins    }  & \num{26}  & \num{45704}   & \num{278}    & \num{20875097} & \num{45703}        & \num{67809}                      & \num{3108339}  & \num{245}  \\
\textsc{pc-sources     }  & \num{231} & \num{307871}  & \num{373}    & \num{11542200} & \num{307871}       & \num{307884}                     & \num{339818}   & \num{792}  \\
\textsc{pcr-cere       }  & \num{6}   & \num{175655}  & \num{3541}   & \num{1446793}  & \num{175643}       & \num{185362}                     & \num{47081}    & \num{535}  \\
\textsc{pcr-einstein.en}  & \num{125} & \num{935920}  & \num{45983}  & \num{49575}    & \num{906995}       & \num{1634034}                    & \num{18192737} & \num{3953} \\
\textsc{pcr-kernel     }  & \num{161} & \num{2755550} & \num{149872} & \num{774532}   & \num{2755550}      & \num{2755556}                    & \num{9258}     & \num{6608} \\
\textsc{pcr-para       }  & \num{6}   & \num{72544}   & \num{2268}   & \num{1926563}  & \num{70680}        & \num{73735}                      & \num{37391}    & \num{265}   \\
			\bottomrule
		\end{tabular}
		}%
	}%
	\caption{Practical evaluation of the algorithm computing all distinct squares on the datasets described in \Cref{secEval}. Execution time is in seconds. It is the median of several conducted experiments, whose variance in time was small. We needed approx.\ 5.73 GB of RAM for each instance.
		The expression $\textup{avg}_{\LCP{}}$ is the average of all \LCP{} values, and $z$ is the number of Lempel-Ziv factors. }
	\label{tableEval}
\end{table}

\section{Decorating the Suffix Tree with All Squares}\label{secDecoration}
\citeauthor{GusfieldS04} described a representation of the set of all distinct squares by a decoration of the suffix tree, like the highlighted nodes (additionally annotated with its respective square) shown in the suffix tree of our running example below.
This representation asks for a set of tuples of the form (node, length) such that
each square~$T[i..i+\ell-1]$ is represented by a tuple~$(v,\ell)$, where $v$ is the highest node 
whose string label has $T[i..i+\ell-1]$ as a (not necessarily proper) prefix.

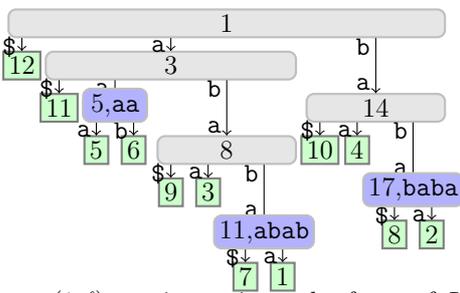
\begin{wrapfigure}{l}{.45\textwidth}
\begin{tikzpicture}[ST,
innernode/.style={rectangle,draw=gray!50,fill=gray!20,thick,minimum width=#1,minimum height=0.5cm,rounded corners=1mm,anchor=south},
]
\expandafter\newcommand\csname nodeType1\endcsname{innernode}
\expandafter\newcommand\csname nodeSize1\endcsname{8.2cm}
\expandafter\newcommand\csname nodeLabel1\endcsname{1}
\coordinate (node1) at (3.85,-0);
\node[innernode=\csname nodeSize1\endcsname] (node 0x11) at (node1) {1};
\expandafter\newcommand\csname nodeType2\endcsname{leafnode}
\expandafter\newcommand\csname nodeSize2\endcsname{0.5cm}
\expandafter\newcommand\csname nodeLabel2\endcsname{12}
\coordinate (node2) at (0,-0.8);
\node[leafnode=\csname nodeSize2\endcsname] (node 1x0) at (node2) {12};
\draw[->] (0,-0) -- (node 1x0.north);
\coordinate (char2p1) at (0,-0.16);
\node[charlabel] at (char2p1){\verb!$!};
\expandafter\newcommand\csname nodeType3\endcsname{innernode}
\expandafter\newcommand\csname nodeSize3\endcsname{4.7cm}
\expandafter\newcommand\csname nodeLabel3\endcsname{3}
\coordinate (node3) at (2.8,-0.8);
\node[innernode=\csname nodeSize3\endcsname] (node 1x8) at (node3) {3};
\draw[->] (2.8,-0) -- (node 1x8.north);
\coordinate (char3p1) at (2.8,-0.16);
\node[charlabel] at (char3p1){\verb!a!};
\expandafter\newcommand\csname nodeType4\endcsname{leafnode}
\expandafter\newcommand\csname nodeSize4\endcsname{0.5cm}
\expandafter\newcommand\csname nodeLabel4\endcsname{11}
\coordinate (node4) at (0.7,-1.6);
\node[leafnode=\csname nodeSize4\endcsname] (node 2x2) at (node4) {11};
\draw[->] (0.7,-0.8) -- (node 2x2.north);
\coordinate (char4p1) at (0.7,-0.96);
\node[charlabel] at (char4p1){\verb!$!};
\expandafter\newcommand\csname nodeType5\endcsname{innernode}
\expandafter\newcommand\csname nodeSize5\endcsname{1.2cm}
\expandafter\newcommand\csname nodeLabel5\endcsname{5}
\coordinate (node5) at (1.75,-1.6);
\node[innernode=\csname nodeSize5\endcsname] (node 2x5) at (node5) {5};
\draw[->] (1.75,-0.8) -- (node 2x5.north);
\coordinate (char5p1) at (1.75,-0.96);
\node[charlabel] at (char5p1){\verb!a!};
\expandafter\newcommand\csname nodeType6\endcsname{leafnode}
\expandafter\newcommand\csname nodeSize6\endcsname{0.5cm}
\expandafter\newcommand\csname nodeLabel6\endcsname{5}
\coordinate (node6) at (1.4,-2.4);
\node[leafnode=\csname nodeSize6\endcsname] (node 3x4) at (node6) {5};
\draw[->] (1.4,-1.6) -- (node 3x4.north);
\coordinate (char6p1) at (1.4,-1.76);
\node[charlabel] at (char6p1){\verb!a!};
\expandafter\newcommand\csname nodeType7\endcsname{leafnode}
\expandafter\newcommand\csname nodeSize7\endcsname{0.5cm}
\expandafter\newcommand\csname nodeLabel7\endcsname{6}
\coordinate (node7) at (2.1,-2.4);
\node[leafnode=\csname nodeSize7\endcsname] (node 3x6) at (node7) {6};
\draw[->] (2.1,-1.6) -- (node 3x6.north);
\coordinate (char7p1) at (2.1,-1.76);
\node[charlabel] at (char7p1){\verb!b!};
\expandafter\newcommand\csname nodeType8\endcsname{innernode}
\expandafter\newcommand\csname nodeSize8\endcsname{2.6cm}
\expandafter\newcommand\csname nodeLabel8\endcsname{8}
\coordinate (node8) at (3.85,-2.4);
\node[innernode=\csname nodeSize8\endcsname] (node 3x11) at (node8) {8};
\draw[->] (3.85,-0.8) -- (node 3x11.north);
\coordinate (char8p1) at (3.85,-0.96);
\node[charlabel] at (char8p1){\verb!b!};
\draw[draw=none] (char8p1) circle [trienode];
\coordinate (char8p2) at (3.85,-1.68);
\node[charlabel] at (char8p2){\verb!a!};
\expandafter\newcommand\csname nodeType9\endcsname{leafnode}
\expandafter\newcommand\csname nodeSize9\endcsname{0.5cm}
\expandafter\newcommand\csname nodeLabel9\endcsname{9}
\coordinate (node9) at (2.8,-3.2);
\node[leafnode=\csname nodeSize9\endcsname] (node 4x8) at (node9) {9};
\draw[->] (2.8,-2.4) -- (node 4x8.north);
\coordinate (char9p1) at (2.8,-2.56);
\node[charlabel] at (char9p1){\verb!$!};
\expandafter\newcommand\csname nodeType10\endcsname{leafnode}
\expandafter\newcommand\csname nodeSize10\endcsname{0.5cm}
\expandafter\newcommand\csname nodeLabel10\endcsname{3}
\coordinate (node10) at (3.5,-3.2);
\node[leafnode=\csname nodeSize10\endcsname] (node 4x10) at (node10) {3};
\draw[->] (3.5,-2.4) -- (node 4x10.north);
\coordinate (char10p1) at (3.5,-2.56);
\node[charlabel] at (char10p1){\verb!a!};
\expandafter\newcommand\csname nodeType11\endcsname{innernode}
\expandafter\newcommand\csname nodeSize11\endcsname{1.2cm}
\expandafter\newcommand\csname nodeLabel11\endcsname{11}
\coordinate (node11) at (4.55,-4);
\node[innernode=\csname nodeSize11\endcsname] (node 5x13) at (node11) {11};
\draw[->] (4.55,-2.4) -- (node 5x13.north);
\coordinate (char11p1) at (4.55,-2.56);
\node[charlabel] at (char11p1){\verb!b!};
\draw[draw=none] (char11p1) circle [trienode];
\coordinate (char11p2) at (4.55,-3.28);
\node[charlabel] at (char11p2){\verb!a!};
\expandafter\newcommand\csname nodeType12\endcsname{leafnode}
\expandafter\newcommand\csname nodeSize12\endcsname{0.5cm}
\expandafter\newcommand\csname nodeLabel12\endcsname{7}
\coordinate (node12) at (4.2,-4.8);
\node[leafnode=\csname nodeSize12\endcsname] (node 6x12) at (node12) {7};
\draw[->] (4.2,-4) -- (node 6x12.north);
\coordinate (char12p1) at (4.2,-4.16);
\node[charlabel] at (char12p1){\verb!$!};
\expandafter\newcommand\csname nodeType13\endcsname{leafnode}
\expandafter\newcommand\csname nodeSize13\endcsname{0.5cm}
\expandafter\newcommand\csname nodeLabel13\endcsname{1}
\coordinate (node13) at (4.9,-4.8);
\node[leafnode=\csname nodeSize13\endcsname] (node 6x14) at (node13) {1};
\draw[->] (4.9,-4) -- (node 6x14.north);
\coordinate (char13p1) at (4.9,-4.16);
\node[charlabel] at (char13p1){\verb!a!};
\expandafter\newcommand\csname nodeType14\endcsname{innernode}
\expandafter\newcommand\csname nodeSize14\endcsname{2.6cm}
\expandafter\newcommand\csname nodeLabel14\endcsname{14}
\coordinate (node14) at (6.65,-1.6);
\node[innernode=\csname nodeSize14\endcsname] (node 2x19) at (node14) {14};
\draw[->] (6.65,-0) -- (node 2x19.north);
\coordinate (char14p1) at (6.65,-0.16);
\node[charlabel] at (char14p1){\verb!b!};
\draw[draw=none] (char14p1) circle [trienode];
\coordinate (char14p2) at (6.65,-0.88);
\node[charlabel] at (char14p2){\verb!a!};
\expandafter\newcommand\csname nodeType15\endcsname{leafnode}
\expandafter\newcommand\csname nodeSize15\endcsname{0.5cm}
\expandafter\newcommand\csname nodeLabel15\endcsname{10}
\coordinate (node15) at (5.6,-2.4);
\node[leafnode=\csname nodeSize15\endcsname] (node 3x16) at (node15) {10};
\draw[->] (5.6,-1.6) -- (node 3x16.north);
\coordinate (char15p1) at (5.6,-1.76);
\node[charlabel] at (char15p1){\verb!$!};
\expandafter\newcommand\csname nodeType16\endcsname{leafnode}
\expandafter\newcommand\csname nodeSize16\endcsname{0.5cm}
\expandafter\newcommand\csname nodeLabel16\endcsname{4}
\coordinate (node16) at (6.3,-2.4);
\node[leafnode=\csname nodeSize16\endcsname] (node 3x18) at (node16) {4};
\draw[->] (6.3,-1.6) -- (node 3x18.north);
\coordinate (char16p1) at (6.3,-1.76);
\node[charlabel] at (char16p1){\verb!a!};
\expandafter\newcommand\csname nodeType17\endcsname{innernode}
\expandafter\newcommand\csname nodeSize17\endcsname{1.2cm}
\expandafter\newcommand\csname nodeLabel17\endcsname{17}
\coordinate (node17) at (7.35,-3.2);
\node[innernode=\csname nodeSize17\endcsname] (node 4x21) at (node17) {17};
\draw[->] (7.35,-1.6) -- (node 4x21.north);
\coordinate (char17p1) at (7.35,-1.76);
\node[charlabel] at (char17p1){\verb!b!};
\draw[draw=none] (char17p1) circle [trienode];
\coordinate (char17p2) at (7.35,-2.48);
\node[charlabel] at (char17p2){\verb!a!};
\expandafter\newcommand\csname nodeType18\endcsname{leafnode}
\expandafter\newcommand\csname nodeSize18\endcsname{0.5cm}
\expandafter\newcommand\csname nodeLabel18\endcsname{8}
\coordinate (node18) at (7,-4);
\node[leafnode=\csname nodeSize18\endcsname] (node 5x20) at (node18) {8};
\draw[->] (7,-3.2) -- (node 5x20.north);
\coordinate (char18p1) at (7,-3.36);
\node[charlabel] at (char18p1){\verb!$!};
\expandafter\newcommand\csname nodeType19\endcsname{leafnode}
\expandafter\newcommand\csname nodeSize19\endcsname{0.5cm}
\expandafter\newcommand\csname nodeLabel19\endcsname{2}
\coordinate (node19) at (7.7,-4);
\node[leafnode=\csname nodeSize19\endcsname] (node 5x22) at (node19) {2};
\draw[->] (7.7,-3.2) -- (node 5x22.north);
\coordinate (char19p1) at (7.7,-3.36);
\node[charlabel] at (char19p1){\verb!a!};

  \ChangeNodeLabel{5}{mast}{5,\texttt{aa}}
  \ChangeNodeLabel{17}{mast}{17,\texttt{baba}}
  \ChangeNodeLabel{11}{mast}{11,\texttt{abab}}

\end{tikzpicture}
\vspace{-1.5em}
\end{wrapfigure}
We show that we can compute this set of tuples in linear time by applying the Phase~II algorithm described in \Cref{secDistinctSquares} to our computed set of all distinct squares.
The Phase~II algorithm takes a list~$L_i$ storing squares starting at text position~$i$, for each $1 \le i \le n$.
Each of these lists has to be sorted in descending order with respect to the squares' lengths.
It is easy to adapt our algorithm to produce these lists:
On reporting a square~$(i,\ell)$, we insert it at the front of~$L_i$.
By doing so, we can fill the lists \emph{without} sorting, since we iterate over the period length in the outer loop, 
while we iterate over all Lempel-Ziv factors in the inner loop.

Finally, we can conduct Phase~II\@.
In the original version, the goal of Phase~II was to decorate the suffix tree with the endpoints of a subset of the original leftmost covering set.
We will show that performing exactly the same operations with the set of the leftmost occurrences of all squares
will decorate the suffix tree with all squares directly.
In more detail, we first augment the suffix tree leaf having label $i$ with the list~$L_i$, for each $1 \le i \le n$.
Subsequently, we follow~\citet{GusfieldS04} by processing every node of the suffix tree with a bottom-up traversal.
During this traversal we propagate the lists of squares from the leaves up to the root:
An internal node~$u$ inherits the list of the child whose subtree contains the leaf with the smallest label among all leaves in the subtree rooted at~$u$.
If the edge to the parent node contains the ending position of one or more squares in the list (these candidates are stored at the front of the list), 
we decorate the edge with these squares, and pop them off from the list.
By~\cite[Theorem 8]{GusfieldS04}, there is no square of the set~$L'$ (defined in \Cref{secDistinctSquares}) neglected during the bottom-top traversal.
The same holds if we exchange~$L'$ with our computed set of all distinct squares:

\begin{Lemma}
	By feeding the algorithm of Phase~II with the above constructed lists~$L_i$ containing the leftmost occurrences of the squares starting at the text position~$i$,
	it will decorate the suffix tree with all distinct squares.
\end{Lemma}
\begin{proof}
We adapt the algorithm of \Cref{secDistinctSquares} to build the lists~$L_i$.
These lists contain the leftmost occurrences of all squares.
	In the following we show that no square is left out during the bottom-up traversal.
	Let us take a suffix tree node~$u$ with its children~$v$ and~$w$.
	Without loss of generality, assume that the smallest label among all leaves contained in the subtree of~$v$ is smaller than 
	the label of every leaf contained in $w$'s subtree.
	For the sake of contradiction, assume that the list of~$w$ contains the occurrence of a square~$(i, \ell)$ at the time when we pass the list of~$v$ to its parent~$u$.
	The length $\ell$ is smaller than $v$'s string depth, otherwise it would already have been popped off from the list. 
	But since $v$'s subtree contains a leaf whose label~$j$ is the smallest among all labels contained in the subtree of~$w$,
	the square occurs before at $T[j..j+\ell-1] = T[i..i+\ell-1]$, a contradiction to the distinctness.
\end{proof}

This concludes the correctness of the modified algorithm. We immediately get:

\begin{theorem}
Given \LPF{}, an LCE data structure on the reversed text, and the suffix tree of~$T$, 
	we can decorate the suffix tree with all squares of the text in $\Oh{n \timeLCE}$ time.
	Asides from these data structures, we use $(\occ + n) \lg n + z \lg z + \min(n + \oh{n}, z \lg n) + \Oh{\lg n}$ bits of working space.
\end{theorem}
\begin{proof}
	We need $(\occ + n) \lg n$ bits for storing the lists~$L_i$ ($\occ \lg n$ bits for storing the lengths of all squares in an integer array, and $n \lg n$ bits for the pointers to the first element of each list).
An LCE query on the text can be answered by the string depth of a lowest common ancestor in the suffix tree; most representations can answer string depth and lowest ancestor queries in constant time.
The array~$Z$ uses $z \lg z$ bits.
The Lempel-Ziv factors are represented as in~\Cref{corLZ}.
\end{proof}

\begin{Corollary}
We can compute the suffix tree and decorate it with all squares of the text in \Oh{n/\epsilon} time using
$(3n+ \occ +2n\epsilon) \lg n + z \lg z + \Oh{n}$ bits, for a constant $0 < \epsilon \le 1$.
\end{Corollary}
\begin{proof}
We use~\Cref{lemLPFmy} to store \SA{}, \ISA{}, \LCP{}, and \LPF{} in $(1+\epsilon) n \lg n + \Oh{n}$ bits.
Subsequently, we build an RMQ data structure on \LCP{} such that LCE queries can be answered in constant time.
	We additionally need the suffix array, its inverse, and the LCP array (with an RMQ data structure) of the reversed text to answer LCE queries on the reversed text.
	Finally, we equip \LPF{} with an RMQ data structure for the right-rotations.
\end{proof}

\JO{%
The values in the lists (i.e., the lengths of the squares starting at a specific position) can be stored in Elias-Fano coding~\cite{Fano,Elias}.
If the list~$L_i$ stores $m_i$ elements, then $2\occ + \sum_{i=1}^n \left( m_i \upgauss{\lg (n/m_i)} \right) + \oh{\occ}$ bits are needed to represent the content of all lists.
It is easy to implement the popping of the first value from a list with this representation, given that we store an offset value and a pointer to the current beginning of each list.
}%

\begin{algorithm}[t]
\DontPrintSemicolon %
$\B{f}$ denotes the left end of a factor~$f = T[\B{f}..\B{f}+\abs{f}-1]$, \lcp{} and \lcs{} compute the LCE in~$T$ and the LCE in the reverse of~$T$ (mirroring the input indices by $i \mapsto n-i$ for $1 \le i \le n-1$), respectively. \;
Let $f_1,\ldots,f_z$ be the factors of the Lempel-Ziv factorization \;
$f_{z+1} \gets T[n]$ \tcp*{dummy factor}

\function{\textup{recursive-rotate}($s$ : \textup{starting position}, $e$: \textup{ending position})}{%
			$m \gets \LPF.\textit{RMQ}[s..e]$ \;
			\lIf{$m > 2p$}{\textbf{return}}
			{report($m, 2p$)} and $\bv{}[m] \gets 1$ \;
			recursive-rotate($s$,$m-1$) and recursive-rotate($m+1$,$e$) \;
}

\function{\textup{right-rotate}($s : $ \textup{starting position of square}, $p$: \textup{period of square})}{%
			\lIf{$\bv{}[s] = 1$}{\textbf{return}}
			\lIf{$\LPF[s] < 2p$}{%
			report($s, 2p$) and $\bv{}[s] \gets 1$
			}
			$\ell \gets \lcp[s, s+p]$ \;
			recursive-rotate($s+1,s+p-1,s+\ell-p$) \;
}

    $Z \gets \text{array of size } z \lg z \text{~bits, zero initialized}$ \;
    $m \gets \max(\abs{f_1}+\abs{f_2},\ldots,\abs{f_{z-1}}+\abs{f_z})$ \;
	\For{$p = 1,\ldots,m$}{%
	$\bv{} \gets \text{bit vector of length~} n$, zero initialized \;
	\For{$x = 1,\ldots,z$}{%
	\If{$\abs{f_x}+\abs{f_{x+1}} < p$}{%
		$y \gets x$ \;
		\While{$\abs{f_y}+\abs{f_{y+1}} < p$}{%
			\lIf{$Z[y] \not= 0$}{%
			$y \gets Z[y]$
			}
			\lElse{\textbf{incr } $y$}
		}
		$Z[x] \gets y$ and
		$x \gets y$ \;
	}%

	\If(\tcp*[h]{probe for squares satisfying \Cref{lemLeftmostOcc}(\ref{itLeftmostOcc1})}){$\abs{f_x} \ge p$}{%
		$q \gets \B{f_{x+1}} - p$ \;
		$\ell_{\textup{R}} \gets \lcp[\B{f_{x+1}}, q]$ and
		$\ell_{\textup{L}} \gets \lcs[\B{f_{x+1}}-1, q-1]$ \;
		\If(\tcp*[h]{found a square of length $2p$ with its right end in $f_{x+1}$}){$\ell_{\textup{R}} + \ell_{\textup{L}} \ge p$ and $\ell_{\textup{R}} > 0$}{%
			$s \gets \max(q - \ell_{\textup{L}}, q - p +1)$ \tcp{square starts at~$s$} 
			right-rotate($s,p$) \;
		}%
	}%
	$q \gets \B{f_x} + p$ \tcp*[h]{probe for squares satisfying \Cref{lemLeftmostOcc}(\ref{itLeftmostOcc2})} \;
	$\ell_{\textup{R}} \gets \lcp[\B{f_x}, q]$ and
	$\ell_{\textup{L}} \gets \lcs[\B{f_x}-1, q-1]$ \;
	$s \gets \max(\B{f_x} - \ell_{\textup{L}}, \B{f_x} - p + 1)$ \tcp{square starts in a factor preceding $f_x$} 
	\If(\tcp*[h]{found a square of length $2p$ whose center is in $f_x$}){$\ell_{\textup{R}} + \ell_{\textup{L}} \ge p$ and $\ell_{\textup{R}} > 0$ and $s + p \le \B{f_{x+1}}$ and $\ell_{\textup{L}} > 0$}{%
			right-rotate($s,p$) \;
	}%
}%
}%
\caption{Modified Algorithm~1 of~\cite{GusfieldS04}}
\label{algo1}
\end{algorithm}

As an application, we consider the common squares problem:
Given a set of non-empty strings with a total length~$n$, we want to find all squares that occur in every string in \Oh{n} time.
We solve this problem by first decorating the generalized suffix tree built on all strings with the distinct squares of all strings.
Subsequently, we apply the \Oh{n} time solution of \citet{Hui92} that annotates each internal suffix tree node~$v$ with the number of strings that contain $v$'s string label.
This solves our problem since we can simply report all squares corresponding to nodes whose string labels are found in all strings.
This also solves the problem asking for the longest common square of all strings in \Oh{n} time, analogously to the longest common substring problem~\cite{gusfield97algorithms}.

Finally, the last section is dedicated to another application of our suffix tree decoration:
\section{Computing the Tree Topology of the MAST in Linear Time}\label{secMAST}
\JO
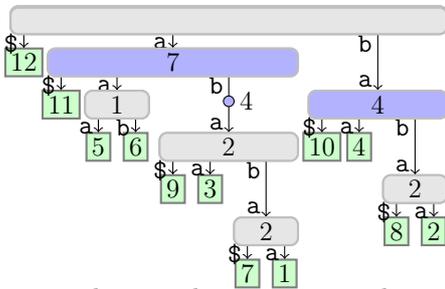
{%
\begin{wrapfigure}{l}{.45\textwidth}
\begin{tikzpicture}[ST,
innernode/.style={rectangle,draw=gray!50,fill=gray!20,thick,minimum width=#1,minimum height=0.5cm,rounded corners=1mm,anchor=south},
]
\expandafter\newcommand\csname nodeType1\endcsname{innernode}
\expandafter\newcommand\csname nodeSize1\endcsname{8.2cm}
\expandafter\newcommand\csname nodeLabel1\endcsname{1}
\coordinate (node1) at (3.85,-0);
\node[innernode=\csname nodeSize1\endcsname] (node 0x11) at (node1) {1};
\expandafter\newcommand\csname nodeType2\endcsname{leafnode}
\expandafter\newcommand\csname nodeSize2\endcsname{0.5cm}
\expandafter\newcommand\csname nodeLabel2\endcsname{12}
\coordinate (node2) at (0,-0.8);
\node[leafnode=\csname nodeSize2\endcsname] (node 1x0) at (node2) {12};
\draw[->] (0,-0) -- (node 1x0.north);
\coordinate (char2p1) at (0,-0.16);
\node[charlabel] at (char2p1){\verb!$!};
\expandafter\newcommand\csname nodeType3\endcsname{innernode}
\expandafter\newcommand\csname nodeSize3\endcsname{4.7cm}
\expandafter\newcommand\csname nodeLabel3\endcsname{3}
\coordinate (node3) at (2.8,-0.8);
\node[innernode=\csname nodeSize3\endcsname] (node 1x8) at (node3) {3};
\draw[->] (2.8,-0) -- (node 1x8.north);
\coordinate (char3p1) at (2.8,-0.16);
\node[charlabel] at (char3p1){\verb!a!};
\expandafter\newcommand\csname nodeType4\endcsname{leafnode}
\expandafter\newcommand\csname nodeSize4\endcsname{0.5cm}
\expandafter\newcommand\csname nodeLabel4\endcsname{11}
\coordinate (node4) at (0.7,-1.6);
\node[leafnode=\csname nodeSize4\endcsname] (node 2x2) at (node4) {11};
\draw[->] (0.7,-0.8) -- (node 2x2.north);
\coordinate (char4p1) at (0.7,-0.96);
\node[charlabel] at (char4p1){\verb!$!};
\expandafter\newcommand\csname nodeType5\endcsname{innernode}
\expandafter\newcommand\csname nodeSize5\endcsname{1.2cm}
\expandafter\newcommand\csname nodeLabel5\endcsname{5}
\coordinate (node5) at (1.75,-1.6);
\node[innernode=\csname nodeSize5\endcsname] (node 2x5) at (node5) {5};
\draw[->] (1.75,-0.8) -- (node 2x5.north);
\coordinate (char5p1) at (1.75,-0.96);
\node[charlabel] at (char5p1){\verb!a!};
\expandafter\newcommand\csname nodeType6\endcsname{leafnode}
\expandafter\newcommand\csname nodeSize6\endcsname{0.5cm}
\expandafter\newcommand\csname nodeLabel6\endcsname{5}
\coordinate (node6) at (1.4,-2.4);
\node[leafnode=\csname nodeSize6\endcsname] (node 3x4) at (node6) {5};
\draw[->] (1.4,-1.6) -- (node 3x4.north);
\coordinate (char6p1) at (1.4,-1.76);
\node[charlabel] at (char6p1){\verb!a!};
\expandafter\newcommand\csname nodeType7\endcsname{leafnode}
\expandafter\newcommand\csname nodeSize7\endcsname{0.5cm}
\expandafter\newcommand\csname nodeLabel7\endcsname{6}
\coordinate (node7) at (2.1,-2.4);
\node[leafnode=\csname nodeSize7\endcsname] (node 3x6) at (node7) {6};
\draw[->] (2.1,-1.6) -- (node 3x6.north);
\coordinate (char7p1) at (2.1,-1.76);
\node[charlabel] at (char7p1){\verb!b!};
\expandafter\newcommand\csname nodeType8\endcsname{innernode}
\expandafter\newcommand\csname nodeSize8\endcsname{2.6cm}
\expandafter\newcommand\csname nodeLabel8\endcsname{8}
\coordinate (node8) at (3.85,-2.4);
\node[innernode=\csname nodeSize8\endcsname] (node 3x11) at (node8) {8};
\draw[->] (3.85,-0.8) -- (node 3x11.north);
\coordinate (char8p1) at (3.85,-0.96);
\node[charlabel] at (char8p1){\verb!b!};
\draw[draw=none] (char8p1) circle [trienode];
\coordinate (char8p2) at (3.85,-1.68);
\node[charlabel] at (char8p2){\verb!a!};
\expandafter\newcommand\csname nodeType9\endcsname{leafnode}
\expandafter\newcommand\csname nodeSize9\endcsname{0.5cm}
\expandafter\newcommand\csname nodeLabel9\endcsname{9}
\coordinate (node9) at (2.8,-3.2);
\node[leafnode=\csname nodeSize9\endcsname] (node 4x8) at (node9) {9};
\draw[->] (2.8,-2.4) -- (node 4x8.north);
\coordinate (char9p1) at (2.8,-2.56);
\node[charlabel] at (char9p1){\verb!$!};
\expandafter\newcommand\csname nodeType10\endcsname{leafnode}
\expandafter\newcommand\csname nodeSize10\endcsname{0.5cm}
\expandafter\newcommand\csname nodeLabel10\endcsname{3}
\coordinate (node10) at (3.5,-3.2);
\node[leafnode=\csname nodeSize10\endcsname] (node 4x10) at (node10) {3};
\draw[->] (3.5,-2.4) -- (node 4x10.north);
\coordinate (char10p1) at (3.5,-2.56);
\node[charlabel] at (char10p1){\verb!a!};
\expandafter\newcommand\csname nodeType11\endcsname{innernode}
\expandafter\newcommand\csname nodeSize11\endcsname{1.2cm}
\expandafter\newcommand\csname nodeLabel11\endcsname{11}
\coordinate (node11) at (4.55,-4);
\node[innernode=\csname nodeSize11\endcsname] (node 5x13) at (node11) {11};
\draw[->] (4.55,-2.4) -- (node 5x13.north);
\coordinate (char11p1) at (4.55,-2.56);
\node[charlabel] at (char11p1){\verb!b!};
\draw[draw=none] (char11p1) circle [trienode];
\coordinate (char11p2) at (4.55,-3.28);
\node[charlabel] at (char11p2){\verb!a!};
\expandafter\newcommand\csname nodeType12\endcsname{leafnode}
\expandafter\newcommand\csname nodeSize12\endcsname{0.5cm}
\expandafter\newcommand\csname nodeLabel12\endcsname{7}
\coordinate (node12) at (4.2,-4.8);
\node[leafnode=\csname nodeSize12\endcsname] (node 6x12) at (node12) {7};
\draw[->] (4.2,-4) -- (node 6x12.north);
\coordinate (char12p1) at (4.2,-4.16);
\node[charlabel] at (char12p1){\verb!$!};
\expandafter\newcommand\csname nodeType13\endcsname{leafnode}
\expandafter\newcommand\csname nodeSize13\endcsname{0.5cm}
\expandafter\newcommand\csname nodeLabel13\endcsname{1}
\coordinate (node13) at (4.9,-4.8);
\node[leafnode=\csname nodeSize13\endcsname] (node 6x14) at (node13) {1};
\draw[->] (4.9,-4) -- (node 6x14.north);
\coordinate (char13p1) at (4.9,-4.16);
\node[charlabel] at (char13p1){\verb!a!};
\expandafter\newcommand\csname nodeType14\endcsname{innernode}
\expandafter\newcommand\csname nodeSize14\endcsname{2.6cm}
\expandafter\newcommand\csname nodeLabel14\endcsname{14}
\coordinate (node14) at (6.65,-1.6);
\node[innernode=\csname nodeSize14\endcsname] (node 2x19) at (node14) {14};
\draw[->] (6.65,-0) -- (node 2x19.north);
\coordinate (char14p1) at (6.65,-0.16);
\node[charlabel] at (char14p1){\verb!b!};
\draw[draw=none] (char14p1) circle [trienode];
\coordinate (char14p2) at (6.65,-0.88);
\node[charlabel] at (char14p2){\verb!a!};
\expandafter\newcommand\csname nodeType15\endcsname{leafnode}
\expandafter\newcommand\csname nodeSize15\endcsname{0.5cm}
\expandafter\newcommand\csname nodeLabel15\endcsname{10}
\coordinate (node15) at (5.6,-2.4);
\node[leafnode=\csname nodeSize15\endcsname] (node 3x16) at (node15) {10};
\draw[->] (5.6,-1.6) -- (node 3x16.north);
\coordinate (char15p1) at (5.6,-1.76);
\node[charlabel] at (char15p1){\verb!$!};
\expandafter\newcommand\csname nodeType16\endcsname{leafnode}
\expandafter\newcommand\csname nodeSize16\endcsname{0.5cm}
\expandafter\newcommand\csname nodeLabel16\endcsname{4}
\coordinate (node16) at (6.3,-2.4);
\node[leafnode=\csname nodeSize16\endcsname] (node 3x18) at (node16) {4};
\draw[->] (6.3,-1.6) -- (node 3x18.north);
\coordinate (char16p1) at (6.3,-1.76);
\node[charlabel] at (char16p1){\verb!a!};
\expandafter\newcommand\csname nodeType17\endcsname{innernode}
\expandafter\newcommand\csname nodeSize17\endcsname{1.2cm}
\expandafter\newcommand\csname nodeLabel17\endcsname{17}
\coordinate (node17) at (7.35,-3.2);
\node[innernode=\csname nodeSize17\endcsname] (node 4x21) at (node17) {17};
\draw[->] (7.35,-1.6) -- (node 4x21.north);
\coordinate (char17p1) at (7.35,-1.76);
\node[charlabel] at (char17p1){\verb!b!};
\draw[draw=none] (char17p1) circle [trienode];
\coordinate (char17p2) at (7.35,-2.48);
\node[charlabel] at (char17p2){\verb!a!};
\expandafter\newcommand\csname nodeType18\endcsname{leafnode}
\expandafter\newcommand\csname nodeSize18\endcsname{0.5cm}
\expandafter\newcommand\csname nodeLabel18\endcsname{8}
\coordinate (node18) at (7,-4);
\node[leafnode=\csname nodeSize18\endcsname] (node 5x20) at (node18) {8};
\draw[->] (7,-3.2) -- (node 5x20.north);
\coordinate (char18p1) at (7,-3.36);
\node[charlabel] at (char18p1){\verb!$!};
\expandafter\newcommand\csname nodeType19\endcsname{leafnode}
\expandafter\newcommand\csname nodeSize19\endcsname{0.5cm}
\expandafter\newcommand\csname nodeLabel19\endcsname{2}
\coordinate (node19) at (7.7,-4);
\node[leafnode=\csname nodeSize19\endcsname] (node 5x22) at (node19) {2};
\draw[->] (7.7,-3.2) -- (node 5x22.north);
\coordinate (char19p1) at (7.7,-3.36);
\node[charlabel] at (char19p1){\verb!a!};

 \draw[mast] (char8p1) circle [trienode] node [trielabel] {4};
 \ChangeNodeLabel{14}{mast}{4}
 \ChangeNodeLabel{3}{mast}{7}

 \ChangeNodeLabel{1}{}{7}
 \ChangeNodeLabel{5}{}{1}
 \ChangeNodeLabel{8}{}{2}
 \ChangeNodeLabel{11}{}{2}
 \ChangeNodeLabel{17}{}{2}
 \ChangeNodeLabel{1}{}{}

\end{tikzpicture}
\vspace{-1.5em}
\end{wrapfigure}
}%
A modification of the suffix tree is the \intWort{minimal augmented suffix tree (MAST)}~\cite{ApostolicoP96}.
This tree can answer the number of the non-overlapping occurrences of a substring in~$T$.
\JO{%
To this end, it adds some nodes on the unary paths of the suffix tree, 
and augments each internal node with the number of the non-overlapping occurrences of its string label, like in the left tree
(the leaves are shown with their suffix number, each leaf represents a substring that occurs exactly once).
The newly created nodes obey the property that the stored numbers of the MAST nodes on the path from a leaf to the root are strictly increasing.
Given a pattern of length~$m$, the MAST can answer the number of the non-overlapping occurrences of the pattern in \Oh{m} time.
To this end, we traverse the MAST from the root downwards while reading the pattern from the edge labels.
If there is a mismatch, the pattern cannot be found in the text.
Otherwise, we end at reading the label of an edge~$(u,v)$; let $u$ be $v$'s parent.
Then the node~$v$ is the highest node whose string label has the pattern as a (not necessarily strict) prefix.
By returning the number stored in $v$ we are done, since this number is the number of non-overlapping occurrences of the pattern in the text.
}%
The MAST can be built in \Oh{n \lg n} time~\cite{solvingstringstatisticsologn}.

In this section, we show how to compute the tree topology of the MAST in linear time.
The topology of the MAST differs to the suffix tree topology by the fact that the root of each square is the string label of a MAST node.
Our goal is to compute a list storing the information about where to insert the missing nodes.
The list stores tuples consisting of a node~$v$ and a length~$\ell$; 
we use this information later to create a new node~$w$ splitting the edge~$(u,v)$ into $(u,w)$ and $(w,v)$,
where $u$ is the (former) parent of~$v$. 
We will label~$(w,v)$ with the last~$\ell$ characters and $(u,v)$ with the rest of the characters of the edge label of $(u,v)$.

To this end, we explore the suffix tree with a top-down traversal while locating the roots of the squares in the order of their lengths.
To locate the roots of the squares in linear time we use two data structures.
The first one is a semi-dynamic lowest marked ancestor data structure~\cite{GabowT85}.
It allows marking a node and querying for the lowest marked ancestor of a node in constant amortized time.
We will use it to mark the area in the suffix tree that has already been processed for finding the roots of the squares.

The second data structure is the list of tuples of the form (node, length) computed in \Cref{secDecoration}, 
where each tuple~$(v,\ell)$ consists of the length~$\ell$ of a square~$T[i..i+\ell-1]$ and the highest suffix tree node~$v$ whose string label
has $T[i..i+\ell-1]$ as a (not necessarily proper) prefix.
We sort this list, which we now call~$L$, with respect to the square lengths with a linear time integer sorting algorithm.

Finally, we explain the algorithm locating the roots of all squares.
We successively process all tuples of~$L$, starting with the shortest square length.
Given a tuple of~$L$ containing the node~$v$ and the length~$\ell$, we 
want to split an edge on the path from the root to~$v$ and insert a new node whose string depth is $\ell/2$.
To this end, we compute the lowest marked ancestor~$u$ of $v$.
If $u$'s string depth is smaller than $\ell/2$,
we mark all descendants of $u$ whose string depths are smaller than $\ell/2$, and additionally the children of those nodes
(this can be done by a DFS or a BFS).
If we query for the lowest marked ancestor of~$u$ again, we get an ancestor~$w$
whose string depth is at least $\ell/2$, and whose parent has a string depth less than $\ell/2$.
We report $w$ and the subtraction of $\ell/2$ from $w$'s string depth (if $\ell/2$ is equal to the string depth of~$w$, then $w$'s string label is equal to the root of $v$'s string label, i.e., we do not have to report it).

\JO{If the suffix tree has a pointer-based representation, it is easy to add the new nodes by splitting each edge~$(u,v)$,
where~$v$ is a node contained in the output list.
}%
\begin{theorem}
We can compute the tree topology of the MAST in linear time using linear number of words.
\end{theorem}
\begin{proof}
	By using the semi-dynamic lowest marked ancestor data structure, 
	we visit a node as many times as we have to insert nodes on the edge to its parent, plus one.
	This gives $\Oh{n + 2 \occ} = \Oh{n}$ time.
\end{proof}

\JO[\block{Open Problems}]{\subsection*{Open Problems}}
It is left open to compute the number of the non-overlapping occurrences of the string labels of the MAST nodes in linear time.
Since RMQ data structures are practically slow, we wonder whether we can avoid the use of any RMQ without loosing linear running time.

\JO[\block{Acknowledgements}]{subsection*{Acknowledgements}}
This work was mainly done during a visit at the Kyushu University in Japan,
supported by the \emph{Japan Society for the Promotion of Science~(JSPS)}.
We thank Thomas Schwentick for the question whether we can run our algorithm online, for which we provided a solution in \Cref{secOnline}.

\bibliographystyle{abbrvnat}
\bibliography{ref}
\newpage
\appendix
\section{Observations}
In~\cite[Line 6 of Algorithm 1b]{GusfieldS04}, the condition $start + k < h_1$ has to be changed to $start + k \le h_1$.
Otherwise, the algorithm would find in $T = {\tt abaabab\$}$ only the square {\tt aa}, but not {\tt abaaba}.

\section{Online Variant}\label{secOnline}
In this section, we consider the \emph{online} setting, where new characters are appended to the end of the text~$T$.
Given the text $T[1..i]$ up to position~$i$ with the Lempel-Ziv factorization~$f_1 \cdots f_y = T$, we consider computing the set of all distinct squares of $f_1 \cdots f_{y-2}$,
i.e., up to the last two Lempel-Ziv factors. 
\newcommand{\timeTrie}[1]{\min\tuple{\lg^2 \lg #1 / \lg \lg \lg #1, \sqrt{\lg #1 /\lg \lg #1}}}
For this setting, we show that we can compute the set of all distinct squares in \Oh{n \timeTrie{n}} time using \Oh{n} words of space.
To this end, we adapt the algorithm of \Cref{thmDistinctSquares} to the online setting.
We need an algorithm computing \LPF{} online, and a semi-dynamic LCE data structure 
(answering LCE queries on the text \emph{and} on the reversed text while supporting appending characters to the text).

The main idea of our solution is to build suffix trees with two online suffix tree construction algorithms.
The first is \citeauthor{ukkonen95online}'s algorithm that computes the suffix tree online in \Oh{n \Time{nav}} time~\cite{ukkonen95online},
		where \Time{nav} is the time for inserting a node and navigating (in particular, selecting the child on the edge starting with a specific character).
	We can adapt this algorithm to compute \LPF{} online:
	Assume that we have computed the suffix tree of $T[1..i-1]$.
	The algorithm processes the new character $T[i]$ by (1) taking the suffix links of the current suffix tree, 
	and (2) adding new leaves where a branching occurs.
	On adding a new leaf with suffix number~$i$, we additionally set $\LPF[i]$ to the string depth of its parent.
	By doing so, we can update the \LPF{} values in time linear to the update time of the suffix tree.
	We build the semi-dynamic RMQ data structure of \citet{fischer11inducing} (or of \cite{UekiDKMNYBIS17} if $n$ is known beforehand) on top of \LPF{}.
	This data structure takes \Oh{n} words and can perform query and appending operations in constant amortized time.

The second suffix tree construction algorithm is a modified version~\cite{BlumerBHECS85} of
\citeauthor{weiner73linear}'s algorithm~\cite{weiner73linear} that builds the suffix tree in the reversed order of \citeauthor{ukkonen95online}'s algorithm in \Oh{n \Time{nav}} time.
Since Weiner's algorithm incrementally constructs the suffix tree of a given text
from right to left, 
we can adapt this algorithm to compute the suffix tree of the reversed text online in \Oh{n \Time{nav}} time.

To get a suffix tree construction time of \Oh{n \timeTrie{n}}, we use the predecessor data structure of~\citet{BeameF02}.
We create a predecessor data structure to store the children of each suffix tree node, 
such that we get the navigation time $\Time{nav} = \Oh{\timeTrie{n}}$ for both suffix trees.
We also create a predecessor data structure to store the out-going suffix links
of each node of the suffix tree constructed by Weiner's algorithm.
Overall, these take a total of \Oh{n} words of space.
	
Finally, our last ingredient is a dynamic lowest common ancestor data structure with \Oh{n} words that performs querying and modification operations in constant time~\cite{cole05dynamic}.
The lowest common ancestor of two suffix tree leaves with the labels~$i$ and~$j$ is the node whose string depth is equal to the longest common extension of~$T[i..]$ and~$T[j..]$, where $T[i..]$ denotes the $i$-th suffix (up to the last position that is available in the online setting).
	Building this data structure on the suffix tree of the text~$T$ and on the suffix tree of the reversed text allows us to compute LCE queries in both directions in constant time.

We adapt the algorithm of \Cref{secDistinctSquares} by switching the order of the loops (again).
The algorithm first fixes a Lempel-Ziv factor~$f_x$ and then searches for squares with a period between one and $\abs{f_x}+\abs{f_{x+1}}$.
Unfortunately, we would need an extra bit vector for each period so that we can track all found leftmost occurrences.
Instead, we use the predecessor data structure of~\cite{BeameF02} storing the found occurrences of squares as pairs of starting positions and lengths.
These pairs can be stored in lexicographic order (first sorted by starting position, then by length).
The predecessor data structure will contain at most~\occ{} elements, hence takes $\Oh{\occ} = \Oh{n}$ words of space.
An insertion and or a search costs us \Oh{\timeTrie{n}} time.

Let us assume that we have computed the set for $T[1..i-1]$, and that the Lempel-Ziv factorization of $T[1..i-1]$ is $f_1 \cdots f_y$.
If appending a new character~$T[i]$ will result in a new factor~$f_{y+1}$, we check
for squares of type \Cref{lemLeftmostOcc}(\ref{itLeftmostOcc1}) and \Cref{lemLeftmostOcc}(\ref{itLeftmostOcc2}) at the borders of $f_{y-1}$.
Duplicates are filtered by the predecessor data structure storing all already reported leftmost occurrences.
The algorithm outputs only the leftmost occurrences with the aid of \LPF{}, whose entries are fixed up to the last two factors 
(this is sufficient since we search for the starting position of the leftmost occurrence of a square with type \Cref{lemLeftmostOcc}(\ref{itLeftmostOcc1}) only in $T[1..\abs{f_1\cdots f_{y-1}}]$, including right-rotations).
In overall, we need $\Oh{(\abs{f_{y-1}}+\abs{f_{y}}) \timeTrie{n}}$ time.

The current bottleneck of the online algorithm is the predecessor data structure in terms of the running time.
Future integer dictionary data structures can improve the overall performance of this algorithm.

\section{Algorithm Execution with one Step at a Time}
In this section, we process the running example $T = ababaaababa\$$ with the algorithm devised in \Cref{secDistinctSquares} step by step.
\SA{}, \LCP{}, \PLCP{}, and \LPF{} are given in the table below:

	\begin{center}
		\tabcolsep=0.1em
			\begin{tabular}{l*{15}{C{1.8em}}}
					$i$        & 1  & 2  & 3 & 4 & 5 & 6 & 7 & 8 & 9  & 10 & 11 & 12
					\\\toprule
					$T$        & {\tt a}  & {\tt b}  & {\tt a} & {\tt b} & {\tt a} & {\tt a} & {\tt a} & {\tt b} & {\tt a}  & {\tt b}  & {\tt a}  & {\tt \$}
					\\\midrule
					$\SA$      & 12 & 11 & 5 & 6 & 9 & 3 & 7 & 1 & 10 & 4  & 8  & 2
					\\\midrule
					$\LCP$     & 0  & 0  & 1 & 2 & 1 & 3 & 3 & 5 & 0  & 2  & 2  & 4
					\\\midrule
					$\PLCP$    & 5  & 4  & 3 & 2 & 1 & 2 & 3 & 2 & 1  & 0  & 0  & 0
					\\\midrule
					$\LPF$     & 0  & 0  & 3 & 2 & 1 & 2 & 5 & 4 & 3  & 2  & 1  & 0
					\\\midrule
					LZ & 
					\multicolumn{1}{c|}{$f_1$}
					&
					\multicolumn{1}{|c|}{$f_2$}
					&
					\multicolumn{3}{|c|}{$f_3$}
					&
					\multicolumn{2}{|c|}{$f_4$}
					&
					\multicolumn{4}{|c|}{$f_5$}
					&
					\multicolumn{1}{|c|}{$f_6$}
					\\\bottomrule
				\end{tabular}
	\end{center}

			\noindent
			The text 
			$T = $
	\RubyReset{}%
	\RubyCount{a}|%
	\RubyCount{b}|%
	\RubyCount{aba}|%
	\RubyCount{aa}|%
	\RubyCount{baba}|%
	\RubyCount{\$} $= f_1 \cdots f_6$
			is factorized in six Lempel-Ziv factors.
			We call $T[1+\abs{f_1 \cdots f_{i-1}}]$ (first position of the $i$-th factor) and 
			$T[1+\abs{f_1 \cdots f_i}]$ (position after the $i$-th factor) the \intWort{left border} and the \intWort{right border} of $f_i$, respectively.
			The idea of the algorithm is to check the presence of a square at a factor border and at an offset value~$q$ of the border with LCE queries.
			$q$ is either the \emph{addition} of $p$ to the \emph{left} border, or the \emph{subtraction} of $p$ from the \emph{right} border.

			\noindent
			The algorithm finds the leftmost occurrences of all squares in the order (first) of their lengths and (second) of their starting positions.
We start with the period~$p=1$ and try to detect squares at each Lempel-Ziv factor border.
To this end, we create a bit vector~\bv{} marking all found squares with period~$p$.
A square of this period is found at the right border of $f_3$.
It is of type \Cref{lemLeftmostOcc}(\ref{itLeftmostOcc1}), since its starting position is in $f_3$.
To find it, we take the right border~$b = 6$ of $f_3$, and the position~$q := b-p = 5$.
We perform an LCE query at~$b$ and~$q$ in the forward and backward direction.
Only the forward query returns the non-zero value of one.
But this is sufficient to find the square {\tt aa} of period one.
Its \LPF{} value is smaller than $2p=2$, so it is the leftmost occurrence.
It is not yet marked in~\bv{}, thus we have not yet reported it.
Right-rotations are not necessary for period~$1$.
Having found all squares with period~$1$, we clear \bv{}.

Next, we search for squares with period~$2$.
We find a square of type \Cref{lemLeftmostOcc}(\ref{itLeftmostOcc2}) at the left border~$b = 2$ of $f_2$.
To this end, we perform an LCE query starting from $b$ and $q := b+p = 4$ in both directions.
Both LCE queries show that $T[1..5]$ is a repetition with period~$p=2$.
Thus we know that $T[1..4]$ is a square.
It is not yet marked in~\bv{}, and has an \LPF{} value smaller than $2p = 4$, i.e., it is a not yet reported leftmost occurrence.
On finding a leftmost occurrence of a square, we right-rotate it, and report all right-rotations whose \LPF{} values are below $2p$.
This is the case for $T[2..5]$, which is the leftmost occurrence of the square {\tt baba}.

After some unsuccessful checks at the next factor borders, we come to factor~$f_5$ and search for a square of type \Cref{lemLeftmostOcc}(\ref{itLeftmostOcc2}).
Two LCE queries in both directions at the left border~$b = 8$ of $f_5$ and $q := b+p = 10$ reveal that $T[7..11]$ is a repetition of period~$2$.
The substring~$T[7..10]$ is a square, but its \LPF{} value is $5 (\ge 2p)$, i.e., we have already reported this square.
Although we have already reported it, some right-rotation of it might not have been reported yet (see~\Cref{secWhyRMQ} for an example).
This time, all right-rotations (i.e., $T[8..12]$) have an \LPF{} value $\ge 2p$, i.e., 
there is no leftmost occurrence of a square of period~$2$ found by right-rotations.
In overall, we have found and reported the leftmost occurrences of all squares \emph{once}.

\section{Need for RMQ on LPF}\label{secWhyRMQ}
In \Cref{secDistinctSquares}, we perform the right-rotations of a square~$(s,2p)$ with an RMQ on the interval $I := [s+1..\min(s+p-1,e-2p+1)]$, where
$e$ is the last position of the maximal repetition of period~$p$ that contains the square. 
Instead of an RMQ, we can linearly scan all \LPF{} values in $I$, giving $\Oh{p} = \Oh{n}$ time.
We cannot do better since the \LPF{} values can be arbitrary.
For instance, consider the text $T = {\tt abaaabaababaaabaaa\$}$.
The text aligned with \LPF{} is shown in the table below.

\noindent
\begin{center}
	\tabcolsep=0.1em
		\begin{tabular}{l*{19}{R{1.7em}}}
				$i$        & 1 & 2 & 3 & 4 & 5 & 6 & 7 & 8 & 9 & 10 & 11 & 12 & 13 & 14 & 15 & 16 & 17 & 18 & 19
				\\\toprule
				$T$        & {\tt a} & {\tt b} & {\tt a} & {\tt a} & {\tt a} & {\tt b} & {\tt a} & {\tt a} & {\tt b} & {\tt a}  & {\tt b}  & {\tt a}  & {\tt a}  & {\tt a}  & {\tt b}  & {\tt a}  & {\tt a}  & {\tt a}  & {\tt \$}
				\\\midrule
				$\LPF$     & 0 & 0 & 1 & 2 & 4 & 3 & 4 & 3 & 2 & 8  & 7  & 6  & 5  & 5  & 4  & 3  & 2  & 1  & 0
				\\\bottomrule
			\end{tabular}
\end{center}

\noindent
The square {\tt abaaabaa} has two occurrences starting at the positions~1 and~10.
The square {\tt baaabaaa} at position~11 is found by right-rotating the occurrence of {\tt abaaabaa} at position~10.
It is found by a linear scan over \LPF{} or an RMQ on \LPF{}.
A slight modification of this example can change the \LPF{} values around this occurrence. 
This shows that we cannot perform a shortcut in general (like stopping the search when the \LPF{} value is at least twice as large as~$p$).

\section{More Evaluation}

\begin{table}[H]
	\centerline{%
		\begin{tabular}{l*{5}{r}}
collection               & 1MiB  & 10MiB & 50MiB  & 100MiB  & 200MiB \\
			\toprule
			\textsc{pc-dblp.xml    } & 0.2 & 3  & 16  & 33   & \num{70}   \\
			\textsc{pc-dna         } & 0.3 & 3  & 23  & 56   & \num{310}  \\
			\textsc{pc-english     } & 0.2 & 5  & 42  & 500  & \num{2639} \\
			\textsc{pc-proteins    } & 0.3 & 4  & 25  & 74   & \num{245}  \\
			\textsc{pc-sources     } & 0.2 & 3  & 31  & 286  & \num{792}  \\
			\textsc{pcr-cere       } & 0.6 & 6  & 30  & 79   & \num{535}  \\
			\textsc{pcr-einstein.en} & 0.4 & 12 & 83  & 1419 & \num{3953} \\
			\textsc{pcr-kernel     } & 0.2 & 8  & 233 & 1274 & \num{6608} \\
			\textsc{pcr-para       } & 0.4 & 4  & 26  & 98   & \num{265}   \\
			\bottomrule
		\end{tabular}
	}%
	\caption{Running times in seconds, evaluated on different input sizes. We took prefixes of 1MiB, 10MiB, and 100MiB of all collections.}
	\label{tableEvalSizes}
\end{table}

\section{Pseudo Code}

\end{document}